\documentclass[11pt]{article}
\usepackage{fullpage}
\usepackage{amsmath,amsthm,amsfonts,dsfont}
\usepackage{amssymb,latexsym,graphicx}
\usepackage{mathtools}
\usepackage{hyperref}
\usepackage{xcolor}
\hypersetup{
	colorlinks,
	linkcolor={red!75!black},
	citecolor={blue!75!black},
	urlcolor={blue!75!black}
}

\newtheorem{theorem}{Theorem}[section]

\newtheorem{proposition}[theorem]{Proposition}
\newtheorem{lemma}[theorem]{Lemma}

\newtheorem{claim}[theorem]{Claim}

\newtheorem{corollary}[theorem]{Corollary}

\newcommand{\C}{\ensuremath{\mathbb{C}}}

\newcommand{\R}{\ensuremath{\mathbb{R}}}
\newcommand{\Z}{\ensuremath{\mathbb{Z}}}

\newcommand{\lat}{\mathcal{L}}
\newcommand{\M}{\mathcal{M}}

\newcommand{\eps}{\varepsilon} 
\renewcommand{\epsilon}{\varepsilon}
\newcommand{\poly}{\mathrm{poly}}

\DeclareMathOperator{\dist}{dist}

\DeclarePairedDelimiter\inner{\langle}{\rangle}

\DeclarePairedDelimiter\floor{\lfloor}{\rfloor}

\mathtoolsset{centercolon}

\newcommand\ket[1]{{ |{#1} \rangle }}

\usepackage[hyperpageref]{backref}
\usepackage{nicematrix}

\begin{document}
	
		\title{\textbf{An Efficient Quantum Factoring Algorithm}}
		\author{
			Oded Regev\thanks{Courant Institute of Mathematical Sciences, New York
				University. Supported by a Simons Investigator Award from the Simons Foundation.}\\
			%
		}
		\date{}
		\maketitle

    \begin{abstract}
    We show that $n$-bit integers can be factorized by independently running a quantum circuit with $\tilde{O}(n^{3/2})$ gates
    for $\sqrt{n}+4$ times, and then using polynomial-time classical post-processing. The correctness of the algorithm relies on a number-theoretic heuristic assumption reminiscent of those used in subexponential classical factorization algorithms.
    It is currently not clear if the algorithm can lead to improved physical implementations in practice.  
    \end{abstract}
  
\section{Introduction}\label{sec:intro}

Shor's celebrated algorithm~\cite{Shor94} allows to factorize $n$-bit integers using a quantum circuit of size (i.e., number of gates)  $\tilde{O}(n^2)$. For factoring to be feasible in practice, however, it is desirable to reduce this number further. Indeed, all else being equal, the fewer quantum gates there are in a circuit, the likelier it is that it can be implemented without noise and decoherence destroying the quantum effects. 

Here we show that quantum circuits of size $\tilde{O}(n^{3/2})$ are enough. More precisely, we present an algorithm that independently runs $\sqrt{n}+4$ times a quantum circuit with $\tilde{O}(n^{3/2})$ gates. The outputs are then classically post-processed in polynomial time (using a lattice reduction algorithm) to generate the desired factorization. 

The quantum circuit size can be made even smaller if super-polynomial-time classical post-processing is allowed. Specifically, for any $0 < \eps \le 1/2$, it can be brought down to $\tilde{O}(n^{3/2-\eps})$ using classical post-processing (solving a hard lattice problem) running in time $\exp(\tilde{O}(n^{2 \eps}))$. The number of times the quantum circuit needs to be applied is still small ($n^{1/2+\eps}$). A curious corollary is that if lattice-based cryptography is broken classically (more precisely, if a polynomial-time classical algorithm exists for hard lattice problems), then quantum circuits of nearly-linear size $\tilde{O}(n)$ are sufficient for factoring integers. This is obtained by taking $\eps=1/2$ in the previous discussion.

A few remarks are in order. First, our algorithm relies on a number-theoretic heuristic assumption reminiscent of those used in subexponential classical factorization algorithms~\cite{CrandallP05}. 
Second, the number of qubits in our quantum circuit is $O(n^{3/2})$, higher than the $O(n)$ in optimized implementations of Shor's algorithm. Third, the quantum circuit's depth is smaller than the one in Shor's original algorithm by $\tilde{O}(n^{1/2})$. 
(Note, though, that Shor's algorithm can be implemented using circuits of depth only $O(\log n)$ at the expense of a (much) larger number of gates $\Omega(n^5)$~\cite{CleveW00}.) 


It is important to keep in mind that our analysis is asymptotic, and hidden constants might make our algorithm inefficient for small values of $n$, such as 2048 bits. 
It remains to be seen whether the algorithm, once sufficiently optimized, can provide an improvement in practice over Shor's algorithm for small $n$. While such an analysis is beyond the scope of this paper, we can make a few remarks. 
First, fast integer multiplication is not advantageous for small $n$, and one instead uses (highly optimized variants of) naive school book integer multiplication, leading to an asymptotic number of gates of approximately $n^3$ for Shor's algorithm and $n^{2.5}$ for our algorithm.
Second, depending on the quantum computer architecture used, the space (or number of qubits) used by the algorithm can play a big role in its practicality and needs to be taken into account. Shor's algorithm is amenable to extensive optimizations, allowing implementations with a very small number of qubits (see~\cite{GidneyE2021} and references therein). It is currently not clear if our algorithm can benefit from all these optimizations, the main issue being our use of repeated squaring (but see below for followup work~\cite{RagavanV23}). 
Finally, lattice reduction algorithms provide surprisingly good approximation factors in practice (e.g., $1.01^d$ for a $d$-dimensional lattice~\cite{GamaN08}), suggesting that our approach can potentially achieve circuit size closer to $n^2$ (without fast integer multiplication), at the cost of increasing the time spent on classical post-processing. 


\paragraph{Statement of the result:}
Fix some $n$-bit number $N \le 2^n$ to be factorized. For some $d>0$, let $b_1,\ldots,b_d$ be some small $O(\log d)$-bit integers (say, $b_i$ is the $i$th prime number) and let $a_i = b_i^2$. Define the lattice
\begin{align}\label{eq:defofl}
\lat &= 
\Big\{ (z_1,\ldots,z_d) \in \Z^d ~\Big|~ \Big(\prod_i b_i^{z_i}\Big)^2 = 1 \bmod N \Big\}
\nonumber \\ 
&= 
\Big\{ (z_1,\ldots,z_d) \in \Z^d ~\Big|~ \prod_i a_i^{z_i} = 1 \bmod N \Big\} \subset \Z^d \; .
\end{align}
Also define the sublattice of $\lat$ given by
\[
\lat_0 = 
\Big\{ (z_1,\ldots,z_d) \in \Z^d ~\Big|~ \prod_i b_i^{z_i} \in \{-1,1\} \bmod N \Big\} \subseteq \lat \; .
\]
Assuming $N$ is odd and not a prime power (since factoring is easy otherwise), we heuristically expect at least half the vectors in $\lat$ to not be in $\lat_0$. For instance, when $N$ is a product of two distinct odd primes, there are 4 square roots of $1$ modulo $N$, so heuristically, half of the vectors in $\lat$ should not be in $\lat_0$. 

Given a vector $z \in \lat \setminus \lat_0$, we have that $b = \prod_i b_i^{z_i} \bmod N$ is a square root of unity modulo $N$ (because $z \in \lat$) yet it is a non-trivial square root of $1$, i.e., not equal to $\pm 1$ modulo $N$ (because $z \notin \lat_0$). In this case, $N$ divides the product $(b-1)(b+1)$ but does not divide either of the terms, and we therefore must have that $\textrm{gcd}(b-1,N)$ is a non-trivial factor of $N$, as desired. Therefore, it suffices to find a vector in $\lat \setminus \lat_0$. 

By the pigeon-hole principle (or Minkowski's first theorem) and using the fact that there are at most $N \le 2^n$ possible values for the product modulo $N$ in Eq.~\eqref{eq:defofl} (i.e., the determinant of $\lat$ is at most $2^n$), $\lat$ is guaranteed to have nonzero vectors of norm at most $\sqrt{d} 2^{n/d}$.\footnote{To see this, consider all vectors $z \in \{-2^{n/d-1}, \ldots , 2^{n/d-1}\}^d$. Since there are more than $2^n \ge N$ such vectors, there are two that lead to the same product in Eq.~\eqref{eq:defofl}. Their difference is therefore in $\lat$, and is of norm at most $\sqrt{d} 2^{n/d}$.} While we expect some (in fact, at least half) of them to be in $\lat \setminus \lat_0$, we do not know how to prove it.\footnote{It was, however, communicated to us by I.\ Shparlinski that assuming the Generalized Riemann Hypothesis, when $N$ is not a prime power and does not have any small prime divisors, $d = \tilde{O}(n^2)$, and $b_i$ is the $i$th prime number, $\lat \setminus \lat_0$ contains a vector with all coordinates in $\{0,1\}$, and in particular, of norm at most $\sqrt{d}$. Unfortunately, this does not seem strong enough for an improved algorithmic result.} Instead, we will make the heuristic assumption that there exists a vector in $\lat \setminus \lat_0$ of norm at most $T=\exp(O(n/d))$. With this assumption, the algorithm is guaranteed to provide a factorization of $N$, as in our main result, stated next. 

\begin{theorem}\label{thm:main}
Let $N$ be an $n$-bit number and assume that for $d=\sqrt{n}$ and $O(\log n)$-bit numbers $b_1,\ldots,b_d$, there exists a vector in $\lat \setminus \lat_0$ of norm at most $T=\exp(O(\sqrt{n}))$. Then, there is a classical polynomial-time algorithm that outputs a non-trivial factor of $N$ using $\sqrt{n}+4$ calls to a quantum circuit of size $O(n^{3/2} \log n)$. 
\end{theorem}

\paragraph{Related work:}
The idea of reducing the cost of the quantum circuit at the expense of applying it several times independently and classically post-processing the
outputs was already suggested by Seifert~\cite{Seifert01} (see also~\cite{EkeraH17,Ekera20}). However, the improvement  obtained by his algorithm is only by a constant factor. We also note that lattices like the one used in our construction were used in different contexts before; see~\cite{DucasP19} and references therein.

\paragraph{Followup work:}
Ragavan and Vaikuntanathan reduced the space requirements of the algorithm to only $O(n \log n)$ qubits~\cite{RagavanV23}, close to the $O(n)$ qubits required in optimized implementations of Shor's algorithm. 
Ekerå and Gärtner introduced several extensions of the algorithm, most notably for solving the discrete logarithm problem~\cite{EkeraG23}. 

\paragraph{Acknowledgements:} The author is grateful to Martin Eker{\aa}, Craig Gidney, Minki Hhan, Igor Shparlinski, Noah Stephens-Davidowitz, Andrew Sutherland, Thomas Vidick, and Ronald de Wolf for their comments on an early draft. 

\section{High level overview of the algorithm}


The algorithm can be thought of as a multidimensional analogue of Shor's algorithm. At the core of the algorithm is the quantum procedure presented in Section~\ref{sec:quantum}. It starts by creating a quantum superposition over $\ket{z}$ for $z \in \Z^d$. For convenience we use a discrete Gaussian state of some radius $R$. It then uses an elementary classical procedure (described below) in superposition to compute $\prod_i a_i^{z_i} \bmod N$ in a new register. This register will be ignored, so effectively can be thought of as being measured. This creates an $\lat$-periodic state over the $\ket{z}$ register, or more precisely, a superposition over a random coset of $\lat$, truncated at radius $R$. We then apply the quantum Fourier transform and measure. This results in vectors from the dual lattice $\lat^*$. However, because the original $\lat$-periodic state only extends up to radius $R$, what we obtain is an \emph{approximation} of vectors in $\lat^*$ up to error roughly $1/R$. The reader might recall a similar phenomenon occurring in Shor's algorithm, and the use of continued fractions there to recover the exact periodicity. In Section~\ref{sec:recovering} we show how to use these approximations to vectors in $\lat^*$ to recover a basis of a lattice $\lat'$ that is essentially the same as $\lat$ for vectors up to norm roughly $2^{-n/d} R$. 
Recalling that we heuristically expect to have vectors in $\lat \setminus \lat_0$ of norm at most $T=\exp(O(n/d))$, we then use the LLL algorithm (which achieves an approximation factor of $2^d$) to obtain a vector in $\lat \setminus \lat_0$ of norm at most $2^{d}\cdot T$. For this to work, we need this norm to be below the bound of $2^{-n/d} R$, i.e., we need $R > 2^{d+n/d} \cdot T$. By taking $d=\sqrt{n}$, it suffices to take $R = \exp(C \sqrt{n})$ for some constant $C>0$. 

As in Shor's algorithm, the quantum circuit size required for this algorithm is dominated by the classical exponentiation procedure. It is here that we benefit from the fact that the $a_i$ are small numbers. Indeed, with $R$ as above, we can compute $\prod_i a_i^{z_i} \bmod N$ using a circuit of size only $\tilde{O}(n^{3/2})$. The elementary but crucial idea is to perform all multiplications on the small numbers $a_i$ directly, so that the only operations we have to perform on large $n$-bit numbers are $\log_2 R = O(\sqrt{n})$ squaring operations. Related algorithmic ideas were used before to compute exponents; see~\cite{Pippenger80} and references therein. 

\section{The quantum procedure}\label{sec:quantum}

Fix $d>0$, $a_1,\ldots,a_d$, and $\lat$ as before. Let $R > \sqrt{2 d}$ be another parameter to be chosen later, and take $D$ to be a power of $2$ somewhat larger than $R$; say, $D \in [2\sqrt{d} \cdot R, 4\sqrt{d} \cdot R) $. We will show a quantum procedure that outputs a uniform sample from $\lat^* / \Z^d$ (a set of cardinality $\det\lat$) perturbed by Gaussian noise of roughly $1/R$ and discretized to the grid $\{0,1/D,\ldots,(D-1)/D\}^d$. 
In more detail, the output will be within statistical distance $1/\poly(d)$ of the distribution $Q$ supported on $\{0,1/D,\ldots,(D-1)/D\}^d$ whose mass at point $w$ is equal to
\begin{align}\label{eq:defofq}
Q(w) := (\det\lat)^{-1} \sum_{v \in \lat^*/\Z^d} Q_v(w) \; ,
\end{align}
where for $v \in \lat^*/\Z^d$, $Q_v$ is the probability distribution defined as
\begin{align}\label{eq:defofqv}
Q_v(w) := \frac{\rho_{1/{\sqrt{2}R}}(v-w+\Z^d)}{\rho_{1/{\sqrt{2}R}}(v-D^{-1} \Z^d)} \; .
\end{align}
Here, we use the Gaussian function $\rho_s: \R^d \to \R$ defined for $s>0$ as
\[
\rho_s(z) = \exp(-\pi \|z\|^2/s^2) \; .
\]
In other words, a sample from $Q$ can be described as the output of the following process: let $v \in \lat^*/\Z^d$ be a uniform coset; then, output a sample from the Gaussian distribution on $\{0,1/D,\ldots,(D-1)/D\}^d$ whose mass at point $w$ is proportional to $\rho_{1/{\sqrt{2}R}}(v-w+\Z^d)$. Importantly, as we show in Claim~\ref{clm:appendix}, with all but probability $O(2^{-d})$, $\dist_{\R^d/\Z^d}(w,v) \le \sqrt{d}/(\sqrt{2}R)$, where $\dist_{\R^d/\Z^d}(w,v):= \min_{z \in \Z^n} \dist(w,v+z)$ denotes distance in the torus, i.e., modulo $1$. 

The quantum procedure starts by approximating to within $1/\poly(d)$ the state proportional to
\begin{align}\label{eq:initialstate}
    \sum_{z \in \{-D/2, \ldots, D/2-1 \}^d}
    \rho_R(z) \ket{z} \; ,
\end{align}
similarly to how it was done in, e.g.,~\cite{Regev09}. Namely, to generate this ``discrete Gaussian state'', first note that it can be written as the tensor product of $d$ copies of the one-dimensional state ($d=1$). It therefore suffices to generate the one-dimensional state. To do this, we use a standard technique which basically puts each qubit from the most-significant to the least-significant into the appropriate superposition of $\ket{0}$ and $\ket{1}$ conditioned on the values of the previous qubits~\cite{GroverR2002}. To obtain a small circuit size, notice that beyond the $O(\log d)$ most significant qubits, all remaining qubits are within distance $1/\poly(d)$ of the ``plus state'' $(\ket{0}+\ket{1})/\sqrt{2}$, no matter what values we condition on. (This uses that $D = O(\poly(d) \cdot R)$ and that $\rho_R$ changes slowly, namely, that for all $z \in \{-D/2,\ldots,D/2-1\}$ and $0 \le k \le D/\poly(d)$, $\rho_R(z)$ is within $1\pm 1/\poly(d)$ of $\rho_R(z+k)$.) We can therefore simply initialize the remaining qubits to the plus state using one Hadamard gate per qubit. In summary, we can approximate the state in~\eqref{eq:initialstate} using a quantum circuit of size  only $d (\log D + \poly(\log d))$, where the $\poly(\log d)$ term is for computing the rotation needed for each of the $O(\log d)$ most significant qubits, and the $\log D$ term is due to the Hadamard gates on the remaining qubits. 

The next step is the most costly one. Here we apply a classical procedure in superposition in order to compute the value $\prod_i a_i^{z_i+D/2} \bmod N$ into a new register $\ket{e}$. (We added $D/2$ for convenience so we do not need to worry about negative exponents.) 
In the following we describe the classical procedure; as is well known, it can be turned into a reversible quantum circuit that cleanly computes the same function. 
Notice that $h(z):=\prod_i a_i^{z_i} \bmod N$ is a homomorphism from $\Z^d$ to $\Z_N^*$, the multiplicative group of integers modulo $N$, and that its kernel is $\lat$. Therefore, there is a bijection between $\Z^d/\lat$ and the image of $h$. As a result, since $\ket{e}$ will be ignored, we can equivalently write the resulting state up to normalization as 
\begin{align}\label{eq:stateafterexponentiation}
   \sum_{e \in \Z^d/\lat}
    \sum_{z \in (\lat+e) \cap [-D/2, D/2)^d}
    \rho_R(z) \ket{z} \ket{e} \; .
\end{align}
To compute $\prod_i a_i^{z_i+D/2} \bmod N$, first notice that when all the exponents are in $\{0,1\}$, we can compute the product of the $d$ numbers in a binary tree fashion, leading to the recurrence $T(d) = 2  T(d/2) + M(d \log d)$, where $M(k)$ is the number of gates needed to compute the product of two $k$-bit numbers, and here we are using that $a_1,\ldots,a_d$ are all small $O(\log d)$ bit numbers. Using fast integer multiplication~\cite{HarveyH21}, $M(k)=O(k \log k)$, leading to a circuit of size $O(d \log^3 d)$. The general case of exponents in $\{0,\ldots,D-1\}$ can be handled using a repeated squaring-like idea. More specifically, for $j=0,\ldots,\lfloor \log_2 (D-1) \rfloor$, let $z_{ij}$ denote the $j$th bit of $z_i+D/2$, with $j=0$ being the most significant. Then, letting $e$ be a register initialized to $1$, we do the following for $j = 0,\ldots,\lfloor \log_2 (D-1) \rfloor$: square $e$, then compute the product of the subset of the $a_{i}$ determined by the $z_{ij}$, and multiply $e$ by the result. To summarize, the circuit size needed for this step is $O(\log D \cdot (d \log^3 d + n \log n))$, where we use that $e$ is an $n$-bit number and can therefore be squared in time $O(n \log n)$.

In the final step, we apply the quantum Fourier transform (over $\Z_D^d$) to the $\ket{z}$ register, and then output the vector in $\{0,1/D,\ldots,(D-1)/D\}^d$ obtained by measuring that register and dividing by $D$. 
As we show using a standard calculation in Appendix~\ref{sec:fourier} (specifically, in Claim~\ref{clm:phioneclosetophitwo} and Proposition~\ref{prop:appendix}), the resulting distribution is within $O(2^{-d})$ distance of the distribution $Q$, as desired. 
The circuit size needed for this step is only 
$O(d \log D \cdot \log( (\log D )/ \eps ))$ 
by using approximate QFT with error $\eps$~\cite{coppersmith2002approximate}. 
Taking $\eps = 1/\poly(d)$, we get circuit size 
$O(d\cdot \log D \cdot (\log \log D  +  \log d ))$.

To summarize, the quantum procedure uses a circuit of size 
\begin{align}\label{eq:quantumcircuitsize}
O(\log D \cdot (d \log^3 d + d \log \log D + n \log n) + d \cdot \poly(\log d))
\end{align}
and outputs a point $w \in [0,1)^d$ that is within distance $\sqrt{d}/(\sqrt{2}R)$ of a uniformly chosen $v \in \lat^* / \Z^d$ (with all but probability $1/\poly(d)$).

\section{Recovering a lattice from noisy samples of the dual}\label{sec:recovering}

\begin{theorem}[\cite{Pomerance01}]\label{thm:pomerance}
Suppose $G$ is a finite abelian group with minimal number of generators $r$. Then, when choosing elements from $G$ independently and uniformly, the expected number of elements needed to generate $G$ is less than $r + \sigma$, where $\sigma = 2.118456563\ldots$. 
\end{theorem}

\begin{corollary}\label{cor:pomerance}
In the setting of Theorem~\ref{thm:pomerance}, $r+4$ uniformly random elements of $G$ generate $G$ with probability at least $1/2$. 
\end{corollary}
\begin{proof}
Otherwise, with probability at least $1/2$, $r+5$ elements are needed to generate $G$. Since this random variable is never smaller than $r$ by assumption, its expectation is at least $r+5/2 > r+\sigma$, in contradiction. 
\end{proof}

\begin{lemma}
\label{lem:randomsamplesdualenough}
Let $\lat \subset \Z^d$, $m \ge d+4$, and assume $v_1,\ldots,v_m$ are uniformly chosen cosets from $\lat^*/\Z^d$. With probability at least $1/4$, it holds that for all nonzero $u \in \Z^d/\lat$, there exists an $i$ such that $\inner{u, v_i} \notin [-\eps,\eps] \bmod 1$, where $\eps = (4 \det \lat)^{-1/m}/3$. 
\end{lemma}
\begin{proof}
First, notice that the group $\lat^*/\Z^d$ can be generated by at most $d$ elements, e.g., by taking a basis of $\lat^*$. Therefore, by Corollary~\ref{cor:pomerance}, with probability at least $1/2$, $v_1,\ldots,v_m$ generate $\lat^*/\Z^d$. Assume that this is the case. Fix some nonzero $u \in \Z^d / \lat$, and consider the distribution of $\inner{u,v} \bmod 1$ where $v$ is uniformly chosen from $\lat^*/\Z^d$. This distribution is not identically zero (as otherwise $u$ would be the zero coset  $\lat$). In fact, it must be equal to the uniform distribution over the set $\{0,1/t,2/t, \ldots, (t-1)/t\}$ for some $t \ge 2$. This follows, e.g., from the invariance of the uniform distribution over $\lat^*/\Z^d$ to shifts by elements from that same group. If $t < 1/\eps$ then by our assumption, there must exist an $i$ such that $\inner{u, v_i} \neq 0 \bmod 1$ which in particular implies that $\inner{u, v_i} \notin [-\eps,\eps] \bmod 1$. Otherwise, assume $t \ge 1/\eps$ and notice that for any fixed $i$, the probability of $\inner{u, v_i} \in [-\eps,\eps] \bmod 1$ is 
\[
(1+2\floor{t \eps })/t
\le 
3\eps \; .
\]
Therefore, the probability that $\inner{u, v_i} \in [-\eps,\eps] \bmod 1$ for all $i \in \{1,\ldots,m\}$ is at most $(3\eps)^m$. We complete the proof by applying the union bound over all $\det \lat - 1$ nonzero elements $u$ in $\Z^d / \lat$.
\end{proof}

\begin{lemma}\label{lem:extendedlattice}
Let $\lat \subset \Z^d$ and $m \ge d+4$. Assume $v_1,\ldots,v_m$ are uniformly chosen cosets from $\lat^*/\Z^d$. For some $\delta>0$ let $w_1,\ldots,w_m \in [0,1)^d$ satisfy that $\dist_{\R^d/\Z^d}(w_i,v_i) < \delta$ for all $i$. For some $S>0$, define the $d+m$-dimensional lattice $\lat'$ generated by the columns of
\NiceMatrixOptions{cell-space-limits = 5pt}
\[
B = 
\begin{pNiceArray}{w{c}{1cm}|c}[margin]
 &  \\
 I_{d \times d} & 0 \\
 &  \\
\hline
 S\cdot w_1 & \\
 \Vdots &  S \cdot I_{m \times m}\\
 S\cdot w_m &  
\end{pNiceArray} \; .
\]
Then, for any $u \in \lat$, there exists a vector $u' \in \lat'$ whose first $d$ coordinates are equal to $u$ and whose norm is at most $\|u\| \cdot (1 + m \cdot S^2 \cdot \delta^2)^{1/2}$. 
Moreover, with probability at least $1/4$ (over the choice of the $v_i$), any nonzero $u' \in \lat'$ of norm $\|u'\| < \min(S,\delta^{-1}) \cdot \eps / 2$ satisfies that its first $d$ coordinates are a nonzero vector in $\lat$, where $\eps = (4 \det \lat)^{-1/m}/3$.
\end{lemma}
\begin{proof}
Take any $u \in \lat$. Then, for any $i$, $\inner{u,v_i} = 0 \bmod 1$ (since $v_i \in \lat^*/\Z^d$), and therefore, by Cauchy-Schwarz, $\inner{u,w_i}$ is within distance $\delta \|u\|$ of an integer. The claim now follows by taking the combination of the first $d$ columns of $B$ given by the coordinates of $u$ and taking an appropriate combination of the remaining $m$ columns of $B$ to make the last $m$ coordinates of the resulting vector at most $S \delta \|u\|$ in absolute value. 
Next, assume $v_1,\ldots,v_m$ satisfy the conclusion of Lemma~\ref{lem:randomsamplesdualenough}, which happens with probability at least $1/4$. Take any nonzero $u' \in \lat'$ and let $u \in \Z^d$ be its first $d$ coordinates. If $u=0$ then clearly $\|u'\| \ge S$, as desired. Assume therefore that $u$ is not in $\lat$, or equivalently, that the coset $u+\lat$ is not the zero element in $\Z^d / \lat$. If $\|u\| \ge \eps / (2 \delta)$, we are done, so assume $\|u\| < \eps / (2 \delta)$. 
By Lemma~\ref{lem:randomsamplesdualenough}, there exists an $i$ such that $\inner{u,v_i}$ is at least $\eps$ away from an integer. By Cauchy-Schwarz, this implies that $\inner{u,w_i}$ is at least
\[
\eps - \delta \|u\| > \eps/2
\]
away from an integer. As a result, $u'$ has a coordinate of absolute value at least $S\eps/2$, as desired. 
\end{proof}

For convenience, we record here the special case of $m=d+4$ and $S=\delta^{-1}$. 

\begin{corollary}\label{cor:extendedlattice}
Let $\lat \subset \Z^d$ and $v_1,\ldots,v_{d+4}$ be uniformly chosen cosets from $\lat^*/\Z^d$. For some $\delta>0$ let $w_1,\ldots,w_{d+4} \in [0,1)^d$ satisfy that $\dist_{\R^d/\Z^d}(w_i,v_i) < \delta$ for all $i$. Let $\lat'$ be the $(2d+4)$-dimensional lattice defined in Lemma~\ref{lem:extendedlattice} (with $S=\delta^{-1}$).
Then, for any $u \in \lat$, there exists a vector $u' \in \lat'$ whose first $d$ coordinates are equal to $u$ and whose norm is at most $\|u\| \cdot (d+5)^{1/2}$. 
Moreover, with probability at least $1/4$ (over the choice of the $v_i$), any nonzero $u' \in \lat'$ of norm $\|u'\| < \delta^{-1} \cdot (4 \det \lat)^{-1/(d+4)}/6$ satisfies that its first $d$ coordinates are a nonzero vector in $\lat$.
\end{corollary}

\section{Proof of the main theorem}

\begin{claim}\label{clm:lll}
There is an efficient classical algorithm that given a basis of a lattice $\lat \subset \R^k$ and some norm bound $T>0$, outputs a list of $\ell \le k$ vectors $z_1,\ldots,z_\ell \in \lat$ of norm at most $\sqrt{k} 2^{k/2} T$ with the property that any vector in $\lat$ of norm at most $T$ must be an integer combination of them. In other words, the sublattice they generate contains all the vectors in $\lat$ of norm at most $T$. 
\end{claim}
\begin{proof}
The algorithm starts by computing an LLL reduced basis of $\lat$, which we denote by $z_1,\ldots,z_k$. It next computes the Gram-Schmidt orthogonalization of the basis, denoted by $\tilde{z}_1, \ldots, \tilde{z}_k$. Finally, let $\ell \ge 0$ be smallest such that $\|\tilde{z}_{\ell+1}\| \ge 2^{k/2} T$ (or $k$ if no such index exists), and output the vectors $z_1,\ldots, z_\ell$. 

By properties of LLL reduced bases, we have that for all $i=1,\ldots,k-1$,  $\|\tilde{z}_{i+1}\| \ge \|\tilde{z}_{i}\| /\sqrt{2}$. It follows that for $i = \ell+1,\ldots,k$, $\|\tilde{z}_{i}\| > T$. As a result, any vector in $\lat$ of norm at most $T$ must be an integer combination of $z_1,\ldots,z_\ell$. 
We complete the proof by recalling that an LLL reduced basis is also ``size reduced,'' implying that
\[
\|z_i\|^2 \le \sum_{j=1}^i \|\tilde{z}_j\|^2 < k 2^{k} T^2 \; .
\]
\end{proof}

\begin{proof}[Proof of Theorem~\ref{thm:main}]
Take $d=\sqrt{n}$ and $R=\exp(C\sqrt{n})$ for a large enough constant $C>0$. With these parameters, and recalling that $D \in [2\sqrt{d} \cdot R, 4\sqrt{d} \cdot R)$, the quantum procedure's circuit size in~\eqref{eq:quantumcircuitsize} becomes
\[
O(n^{3/2} \log n)
\]
and its output is a point $w$ within distance $\delta = \sqrt{d}/(\sqrt{2}R)$ of a uniformly chosen $v \in \lat^* / \Z^d$ (with all but probability $1/\poly(d)$). Apply the quantum procedure $d+4$ times independently to obtain such vectors $w_1,\ldots,w_{d+4}$. 

Consider the $(2d+4)$-dimensional lattice $\lat'$ given in Corollary~\ref{cor:extendedlattice}. By our assumption, there exists a vector in $u' \in \lat'$ of norm at most $(d+5)^{1/2} \cdot T$ whose first $d$ coordinates are a nonzero vector $u \in \lat \setminus \lat_0$. We next apply the classical algorithm in Claim~\ref{clm:lll} to $\lat'$ with the norm bound $(d+5)^{1/2} \cdot T$. As its output, we obtain vectors $z'_1,\ldots,z'_\ell$ of norm at most 
\[
(2d+4)^{1/2} 2^{d+2} (d+5)^{1/2} \cdot T
<
\delta^{-1} (4 \det \lat)^{-1/(d+4)} / 6 \; ,
\]
where the inequality follows since $\det\lat \le N \le 2^n$ and by choosing $R$ large enough. As a result, by the second property in Corollary~\ref{cor:extendedlattice}, except with probability $1/4$, if we denote the first $d$ coordinates of $z'_i$ by $z_i$, we have that $z_i \in \lat$ for all $i$. Moreover, at least one of the $z_i$ must not be in $\lat_0$, otherwise $u$, which is an integer combination of the $z_i$ (since $u'$ is an integer combination of the $z'_i$) would also be in $\lat_0$. Finally, we apply for each of the $z_i$ the gcd calculation outlined in Section~\ref{sec:intro}. Since there exists an $i$ such that $z_i \in \lat \setminus \lat_0$, one of these calculations will yield a non-trivial factor of $N$, as desired.
\end{proof}

The extension to super-polynomial classical post-processing mentioned in the introduction is similar. One needs to take $d = n^{1/2+\eps}$, $R = \exp(C n^{1/2-\eps})$, and use the fact that for all $0 < \delta < 1$, one can approximate lattice problems in dimension $d$ to within $\exp(d^{1-\delta})$ in time $\exp(\tilde{O}(d^\delta))$~\cite{GamaN08}.

%

\appendix

\section{Fourier transform calculation}\label{sec:fourier}

We will need the following useful fact by Banaszczyk. It shows that for any lattice $\lat$, almost all the Gaussian mass in $\rho_s$ is given by points of norm at most $\sqrt{d}s$. Here and elsewhere, for a set $A \subset \R^n$, we use the notation $f(A)$ to denote the sum $\sum_{x \in A} f(x)$. 

\begin{lemma}[\cite{banaszczykNewBoundsTransference1993}]\label{lem:bana1}
For any $d$-dimensional lattice $\lat$, $x \in \R^n$, and $s>0$,
 \[
 \rho_s(\{ y \in \lat+x ~|~ \|y\| > \sqrt{d} s \} ) < 2^{-d} \cdot \rho_s(\lat) \; .  
 \]
\end{lemma}

\begin{corollary}\label{cor:bana}
If $\lat$ is a lattice containing no nonzero vectors of norm at most $\sqrt{d} s$ then 
 $ \rho_s(\lat \setminus \{0\}) \le 2\cdot 2^{-d}$.
\end{corollary}
\begin{proof}
Using Lemma~\ref{lem:bana1} with $x=0$, write
\[
\rho_s(\lat \setminus \{0\}) =
\rho_s(\{ y \in \lat ~|~ \|y\| > \sqrt{d} s \} ) < 
2^{-d} \cdot (1 + \rho_s(\lat \setminus \{0\}) )
\; 
\]
and rearrange.
\end{proof}

We will use the following formulation of the Poisson summation formula.
Here, $\hat{f}$ denotes the Fourier transform of $f$. For instance, $\widehat{\rho_s}=s^n \rho_{1/s}$.
\begin{lemma}(Poisson summation formula)\label{lem:psf}
For any lattice $\lat$ and any (nice enough) function $f: \R^n \to \C$,
\[ 
f(\lat) = \det(\lat^*) \hat{f}(\lat^*) \; ,
\]
where $\hat{f}$ denotes the Fourier transform of $f$.
\end{lemma}

Let $\ket{\varphi_1}$ be the state in Eq.~\eqref{eq:stateafterexponentiation} which, to recall, is given by 
\begin{align*}
\ket{\varphi_1} &= 
   Z_1^{-1} \sum_{e \in \Z^d/\lat}
    \sum_{z \in (\lat+e) \cap \{-D/2, \ldots, D/2-1\}^d}
    \rho_R(z) \ket{z} \ket{e} \; ,
\end{align*}
where $Z_1 > 0$ is the normalization term. 

\begin{claim}\label{clm:zone}
We have that $Z_1^2 \in [ 1\pm 2 \cdot 2^{-d}] (R/\sqrt{2})^d$. 
\end{claim}
\begin{proof}
Notice that
\[
Z_1^2 
= \sum_{z \in \{-D/2, \ldots, D/2-1\}^d} \rho_R(z)^2
= \sum_{z \in \{-D/2, \ldots, D/2-1\}^d} \rho_{R/\sqrt{2}}(z) \; .
\]
Therefore, by Lemma~\ref{lem:bana1} (with $x=0$) and using $D/2 \ge \sqrt{d} R/\sqrt{2}$, $Z_1^2$ satisfies
\[
(1-2^{-d}) \cdot \rho_{R/\sqrt{2}}(\Z^d) \le
Z_1^2
\le \rho_{R/\sqrt{2}}(\Z^d) \; .
\]
By the Poisson summation formula, $\rho_{R/\sqrt{2}}(\Z^d)$ is equal to $(R/\sqrt{2})^d \rho_{\sqrt{2}/R}(\Z^d)$. Moreover, since $1 \ge \sqrt{d} \cdot \sqrt{2} / R$, Corollary~\ref{cor:bana} shows that $\rho_{\sqrt{2}/R}(\Z^d) \in [1,1+2\cdot 2^{-d}]$. 
\end{proof}

Also define the state
\[
\ket{\varphi_2} = 
   Z_2^{-1} \sum_{e \in \Z^d/\lat}
    \sum_{z \in \Z_D^d}
    \rho_R\Big((z+D\Z_d) \cap (\lat+e)  \Big) \ket{z} \ket{e}     
    \; ,
\]
where, again, $Z_2>0$ is the normalization term. In other words, whereas in $\ket{\varphi_1}$ we truncate the discrete Gaussian to the box $\{-D/2,\ldots, D/2-1\}^d$, in $\ket{\varphi_2}$ we let it wrap around modulo $D$. In the next claim, we show that the two states are very close, which intuitively follows from the fact that the Gaussian mass does not extend much beyond radius $\sqrt{d} R$. 

\begin{claim}\label{clm:phioneclosetophitwo}
We have that
\[
\| \ket{\varphi_1} - \ket{\varphi_2} \|_2
\le 2 \cdot 2^{-d} \; .
\]
Moreover, $Z_1/Z_2 \in [1 \pm 2^{-d}]$.
\end{claim}
\begin{proof}
Let $\M$ be the $2d$-dimensional lattice given by all vectors $(z_1,z_2) \in \Z^{2d}$ such that both $z_1 = z_2 \bmod D$ and $z_1 = z_2 \bmod \lat$. Then notice that
\begin{align*}
Z_2^2 &= 
\sum_{e \in \Z^d/\lat}
\sum_{z \in \Z_D^d}
 \rho_R\Big((z+D\Z_d) \cap (\lat+e)  \Big)^2 \\
&=
\sum_{\substack{z_1,z_2 \in \Z^d \\ 
z_1 = z_2 \bmod D \\ 
z_1 = z_2 \bmod \lat}}
 \rho_R(z_1) \cdot \rho_R(z_2) \\
&=
\rho_R(\M) \; .
\end{align*}
Similarly, denoting by $\M'$ the subset of $\M$ corresponding to all vectors $(z_1,z_2)$ such that neither $z_1$ nor $z_2$ are in $\{-D/2,\ldots,D/2-1\}^d$,
\begin{align*}
\| Z_1 \ket{\varphi_1} - Z_2 \ket{\varphi_2} \|_2^2
&= 
   \sum_{e \in \Z^d/\lat}
    \sum_{z \in \Z_D^d}
    \rho_R\Big(((z+D\Z_d) \cap (\lat+e)) \setminus \{-D/2, \ldots, D/2-1\}^d \Big)^2   \\
&=
\rho_R(\M') \\
& \le 
2^{-2d} \cdot \rho_R(\M)  \\
& = 
2^{-2d} \cdot Z_2^2
\; ,
\end{align*}
where we used Lemma~\ref{lem:bana1} (with $x=0$) and the fact that all vectors in $\M'$ are of norm at least $D/\sqrt{2} \ge \sqrt{2d} R$. Dividing both sides by $Z_2^2$, we get
\[
\| (Z_1/Z_2) \ket{\varphi_1} - \ket{\varphi_2} \|_2
\le 
2^{-d} 
\; .
\]
By the triangle inequality, this implies that $Z_1/Z_2 = \| (Z_1/Z_2) \cdot \ket{\varphi_1} \|_2 \in [1 \pm 2^{-d}]$. Using the triangle inequality again, we get that 
\[
\| \ket{\varphi_1} - \ket{\varphi_2} \|_2
\le
\| \ket{\varphi_1} - (Z_1/Z_2) \cdot \ket{\varphi_1} \|_2 +
\| (Z_1/Z_2) \cdot \ket{\varphi_1} - \ket{\varphi_2} \|_2 
\le 2 \cdot 2^{-d} \; ,
\]
as desired. 
\end{proof}

\begin{proposition}\label{prop:appendix}
The distribution obtained by applying QFT to $\ket{\varphi_2}$, discarding the $e$ register, measuring the $z$ register, and dividing the result by $D$ is within statistical distance $O(2^{-d})$ of the distribution $Q$ defined in Eq.~\eqref{eq:defofq}. 
\end{proposition}
\begin{proof}
The QFT of $\ket{\varphi_2}$ is given by
\begin{align*}
   &
   D^{-d/2}\cdot 
   Z_2^{-1} \cdot
   \sum_{e \in \Z^d/\lat}
    \sum_{w \in \{0,1/D,\ldots,(D-1)/D\}^d}
    \sum_{z \in \Z_D^d}
    \exp(2 \pi i \inner{w,z}) 
    \rho_R\Big((z+D\Z_d) \cap (\lat+e)  \Big) 
    \ket{w} \ket{e}      \\
    &=
   D^{-d/2}\cdot 
   Z_2^{-1} \cdot
   \sum_{e \in \Z^d/\lat}
    \sum_{w \in \{0,1/D,\ldots,(D-1)/D\}^d}
    \sum_{z \in \lat + e}
    \exp(2 \pi i \inner{w,z}) 
    \rho_R(z) 
    \ket{w} \ket{e}        
    \;.
\end{align*} 
Using the Poisson summation formula, this is equal to
\begin{align*}
    R^d\cdot 
    D^{-d/2}\cdot 
    (\det \lat)^{-1} \cdot 
    Z_2^{-1}
    \sum_{e \in \Z^d/\lat} 
    \sum_{w \in \{0,1/D,\ldots,(D-1)/D\}^d}
    \sum_{v \in \lat^*}
    \exp(2 \pi i \inner{v,e}) \rho_{1/R}(v-w) 
    \ket{w} \ket{e} \; .
\end{align*}
Therefore, after measuring (and discarding) $e$, the probability of measuring $w$ is
\begin{align*}
    P(w) := &R^{2d}\cdot 
    D^{-d}\cdot 
    (\det \lat)^{-2} \cdot 
    Z_2^{-2} \cdot 
    \sum_{e \in \Z^d/\lat} 
    \Big| \sum_{v \in \lat^*}
     \exp(2 \pi i \inner{v,e}) \rho_{1/R}(v-w) \Big|^2  \; .
\end{align*}
The sum can be simplified as
\begin{align*}
    \sum_{e \in \Z^d/\lat} &
    \sum_{v_1,v_2 \in \lat^*}
     \exp(2 \pi i \inner{v_1-v_2,e}) \rho_{1/R}(v_1-w)\rho_{1/R}(v_2-w) \\
    &=
    \det \lat \cdot 
    \sum_{\substack{v_1,v_2 \in \lat^* \\ v_1 = v_2 \bmod \Z^d}}
     \rho_{1/R}(v_1-w)\rho_{1/R}(v_2-w) \\
    & \ge
    \det \lat \cdot 
    \sum_{v \in \lat^*}
     \rho_{1/R}(v-w)^2 \\
    &=     
    \det \lat \cdot 
    \sum_{v \in \lat^*/\Z^d}
    \rho_{1/(\sqrt{2} R)}(v-w+\Z^d) \; ,
\end{align*}
where in the first equality we used that $\sum_{e \in \Z^d/\lat} \exp(2 \pi i \inner{v_1-v_2,e})$ is equal to $\det \lat$ if $v_1=v_2 \bmod \Z^d$ and is zero otherwise, and the inequality is simply by discarding terms where $v_1 \neq v_2$. 
Therefore, $P(w)$ is bounded from below by
\begin{align*}
    P'(w) &:= 
	R^{2d}\cdot D^{-d}\cdot (\det \lat)^{-1} \cdot Z_2^{-2} \cdot
    \sum_{v \in \lat^*/\Z^d}
    \rho_{1/(\sqrt{2} R)}(v-w+\Z^d) \\
	& = \; 
    \sum_{v \in \lat^*/\Z^d} \alpha_v Q_v(w) \; ,
\end{align*}
where $Q_v$ is the probability distribution defined in Eq.~\eqref{eq:defofqv}, and $\alpha_v$ is given by
\begin{align}
\alpha_v
&=
R^{2d}\cdot D^{-d}\cdot (\det \lat)^{-1} \cdot Z_2^{-2} \cdot
\rho_{1/(\sqrt{2}R)}(v+D^{-1} \Z^d) \nonumber \\
&= 
(R/\sqrt{2})^{d}\cdot (\det \lat)^{-1} \cdot Z_2^{-2} \cdot
\sum_{x \in D \Z^d} \exp(2 \pi i \inner{v,x}) \rho_{\sqrt{2}R}(x) \nonumber  \\
& \in 
(R/\sqrt{2})^{d}\cdot (\det \lat)^{-1} \cdot Z_2^{-2} \cdot
[1 \pm 2 \cdot 2^{-d}] \label{eq:sameprobabilityv}
\; ,
\end{align}
where we used the Poisson summation formula, and 
in the last step we applied Corollary~\ref{cor:bana} (using $D \ge \sqrt{2d}R$) and the triangle inequality. 
From Claim~\ref{clm:zone} and the second part of Claim~\ref{clm:phioneclosetophitwo}, we get that $Z_2^2$ is within $1 \pm O(2^{-d})$ of $(R/\sqrt{2})^d$, and so $\alpha_v \in (\det \lat)^{-1} \cdot [1 \pm O(2^{-d})]$.
This establishes that $P'$ is within $\ell_1$ distance $O(2^{-d})$ of $Q$. In particular, its $\ell_1$ norm (i.e., total mass) is at least $1-O(2^{-d})$, which, combined with it being a lower bound on $P$, implies that it is within $\ell_1$ distance $O(2^{-d})$ of $P$. Using the triangle inequality,
\[
\| P - Q \|_1 \le \|P - P' \|_1 + \|P' - Q\|_1 = O(2^{-d}) \; ,
\]
completing the proof.
\end{proof}

\begin{claim}\label{clm:appendix}
For any $v \in \R^d / \Z^d$, if $w$ is chosen from the distribution $Q_v$ defined in Eq.~\eqref{eq:defofqv}, then the probability that $\dist_{\R^d/\Z^d}(w,v) > \sqrt{d}/(\sqrt{2}R)$
is at most $O(2^{-d})$.
\end{claim}
\begin{proof}
Using Lemma~\ref{lem:bana1}, 
\begin{align*}
\rho_{1/{\sqrt{2}R}}&( \{ x \in v - D^{-1} \Z^d ~|~ \dist_{\R^d/\Z^d}(x,0) > \sqrt{d} / (\sqrt{2} R)\})  \\
&\le
\rho_{1/{\sqrt{2}R}}( \{ x \in v - D^{-1} \Z^d ~|~ \|x\| > \sqrt{d} / (\sqrt{2} R)\})  \\
&< 
2^{-d}
\rho_{1/{\sqrt{2}R}}(D^{-1}\Z^d) \\
&\le 
(1+O(2^{-d})) 2^{-d}
\rho_{1/{\sqrt{2}R}}(v-D^{-1}\Z^d) \; ,
\end{align*}
where in the last inequality we used Eq.~\eqref{eq:sameprobabilityv} (specifically, that $\alpha_0 \le (1+O(2^{-d})) \alpha_v$).
\end{proof}



\bibliographystyle{alpha}
\bibliography{quantum_factoring}

\begin{thebibliography}{HvdH21}

\bibitem[Ban93]{banaszczykNewBoundsTransference1993}
Wojciech Banaszczyk.
\newblock New bounds in some transference theorems in the geometry of numbers.
\newblock {\em Mathematische Annalen}, 296(4):625--635, 1993.

\bibitem[Cop02]{coppersmith2002approximate}
Don Coppersmith.
\newblock An approximate {Fourier} transform useful in quantum factoring, 2002.
\newblock \href{https://arxiv.org/abs/quant-ph/0201067}{quant-ph/0201067}.

\bibitem[CP05]{CrandallP05}
Richard Crandall and Carl Pomerance.
\newblock {\em Prime numbers}.
\newblock Springer, New York, second edition, 2005.
\newblock A computational perspective.

\bibitem[CW00]{CleveW00}
Richard Cleve and John Watrous.
\newblock Fast parallel circuits for the quantum {F}ourier transform.
\newblock In {\em 41st {A}nnual {S}ymposium on {F}oundations of {C}omputer
  {S}cience ({R}edondo {B}each, {CA}, 2000)}, pages 526--536. IEEE Comput. Soc.
  Press, Los Alamitos, CA, 2000.

\bibitem[DP19]{DucasP19}
L\'{e}o Ducas and C\'{e}cile Pierrot.
\newblock Polynomial time bounded distance decoding near {M}inkowski's bound in
  discrete logarithm lattices.
\newblock {\em Des. Codes Cryptogr.}, 87(8):1737--1748, 2019.

\bibitem[EG23]{EkeraG23}
Martin Ekerå and Joel Gärtner.
\newblock Extending regev's factoring algorithm to compute discrete logarithms,
  2023.

\bibitem[EH17]{EkeraH17}
Martin Eker{\aa} and Johan H{\aa}stad.
\newblock Quantum algorithms for computing short discrete logarithms and
  factoring {RSA} integers.
\newblock In {\em Post-quantum cryptography}, volume 10346 of {\em Lecture
  Notes in Comput. Sci.}, pages 347--363. Springer, Cham, 2017.

\bibitem[Eke20]{Ekera20}
Martin Eker\aa.
\newblock On post-processing in the quantum algorithm for computing short
  discrete logarithms.
\newblock {\em Des. Codes Cryptogr.}, 88(11):2313--2335, 2020.

\bibitem[GE21]{GidneyE2021}
Craig Gidney and Martin Eker{\aa{}}.
\newblock How to factor 2048 bit {RSA} integers in 8 hours using 20 million
  noisy qubits.
\newblock {\em {Quantum}}, 5:433, April 2021.

\bibitem[GN08]{GamaN08}
Nicolas Gama and Phong~Q. Nguyen.
\newblock Finding short lattice vectors within {M}ordell's inequality.
\newblock In {\em S{TOC}'08}, pages 207--216. ACM, New York, 2008.

\bibitem[GR02]{GroverR2002}
Lov Grover and Terry Rudolph.
\newblock Creating superpositions that correspond to efficiently integrable
  probability distributions, 2002.
\newblock arXiv:quant-ph/0208112.

\bibitem[HvdH21]{HarveyH21}
David Harvey and Joris van~der Hoeven.
\newblock Integer multiplication in time {$O(n \log n)$}.
\newblock {\em Ann. of Math. (2)}, 193(2):563--617, 2021.

\bibitem[Pip80]{Pippenger80}
Nicholas Pippenger.
\newblock On the evaluation of powers and monomials.
\newblock {\em SIAM J. Comput.}, 9(2):230--250, 1980.

\bibitem[Pom01]{Pomerance01}
Carl Pomerance.
\newblock The expected number of random elements to generate a finite abelian
  group.
\newblock {\em Period. Math. Hungar.}, 43(1-2):191--198, 2001.

\bibitem[Reg09]{Regev09}
Oded Regev.
\newblock On lattices, learning with errors, random linear codes, and
  cryptography.
\newblock {\em J. ACM}, 56(6):Art. 34, 40, 2009.

\bibitem[RV23]{RagavanV23}
Seyoon Ragavan and Vinod Vaikuntanathan.
\newblock Optimizing space in regev's factoring algorithm, 2023.

\bibitem[Sei01]{Seifert01}
Jean-Pierre Seifert.
\newblock Using fewer qubits in {S}hor's factorization algorithm via
  simultaneous {D}iophantine approximation.
\newblock In {\em Topics in cryptology---{CT}-{RSA} 2001 ({S}an {F}rancisco,
  {CA})}, volume 2020 of {\em Lecture Notes in Comput. Sci.}, pages 319--327.
  Springer, Berlin, 2001.

\bibitem[Sho99]{Shor94}
Peter~W. Shor.
\newblock Polynomial-time algorithms for prime factorization and discrete
  logarithms on a quantum computer.
\newblock {\em SIAM Rev.}, 41(2):303--332, 1999.

\end{thebibliography}

\end{document}